%% file: Infinite_Open_Games.tex
\documentclass[preprint,12pt,authoryear]{elsarticle}

\usepackage{amssymb, amsmath, amsthm}
\usepackage{enumerate}
\usepackage{hyperref}
\usepackage{microtype}

\usepackage[all]{xy}

\usepackage{tikz}
\usetikzlibrary{shapes}
\usetikzlibrary{arrows}

\newtheorem{theorem}{Theorem}
\newtheorem{lemma}[theorem]{Lemma}
\newtheorem{prop}[theorem]{Proposition}
\newtheorem{fact}[theorem]{Fact}
\theoremstyle{definition}
\newtheorem{definition}[theorem]{Definition}

\input{macros.tex}

\begin{document}

\journal{Theoretical Computer Science}

\begin{frontmatter}

\title{A Compositional Treatment of Iterated Open Games}

\author[strath]{Neil Ghani}
\author[strath]{Clemens Kupke}
\author[strath]{Alasdair Lambert}
\author[strath]{Fredrik Nordvall Forsberg}

\address[strath]{University of Strathclyde}

\begin{abstract}
  Compositional Game Theory is a new, recently introduced model of
  economic games based upon the computer science idea of
  compositionality. In it, complex and irregular games can be built up
  from smaller and simpler games, and the equilibria of these complex
  games can be defined recursively from the equilibria of their
  simpler subgames. This paper extends the model by providing a final
  coalgebra semantics for infinite games. In the course of this, we
  introduce a new operator on games to model the economic concept of
  {\em subgame perfection}.
\end{abstract}

\end{frontmatter}

\section{Introduction}
\label{sec:intro}

Compositionality, where one sees complex systems as being built from
smaller subsystems, is widely regarded within computer science as best
practice. As the subsystems are smaller, they are easier to
reason about, and compositionality also promotes modularity and reuse;
a particular system can be a subsystem of many different supersystems.
Can compositionality be applied also to economic games? In general,
not all reasoning is compositional, especially if significant emergent
behaviour is present in a large system but not in its subsystems.
This is unfortunately the case for economic games.  For example, if
$\sigma$ is an optimal strategy for $G$, then is $\sigma$ part of an
optimal strategy for $G \ast H$, where $G \ast H$ is a super-game
built from $G$ and $H$?  Clearly not, e.g.\ the {\em Iterated
  Prisoners' Dilemma} has equilibria --- such as cooperative
equilibria --- that do not arise from repeatedly playing the Nash
equilibrium from the {\em Prisoners' Dilemma}
\citep{evolutionOfCooperation}.

However, \citet{compositionalGames} produced a compositional model of
game theory which included a limited set of operators for building new
games from old. There was no operator to compositionally build the
infinite prisoners' dilemma from the one-shot prisoners dilemma and,
more generally, to compositionally build infinite iterations of
games. This paper addresses that problem. Within programming language
theory, these sort of issues are tackled by final coalgebra semantics
\citep{coalgebra} and we follow this practice, with
the added benefit of bringing related bisimulation techniques to the
game theory community. In doing this, we deal with a number of issues:
\begin{itemize}
\item Each round of an infinite game produces utility. Traditionally,
  this infinite sequence of staged utilities is combined into a single
  utility in one of a number of ad hoc manners. We take the bolder
  approach of not requiring the choice of a single mechanism for
  combining utilities.
\item The general approach of Compositional Game Theory deals with a
  new concept of coutility: if utility is gained by one agent, it must
  come from another agent. However, this produces problems for
  modelling infinite games, and so we make a simplifying assumption
  with respect to coutility. This is not a limitation in practice,
  as standard treatments do not consider coutility.
\item The coalgebraic approach we advocate dovetails well with the
  economic concept of {\em subgame perfection} where a strategy must
  be an optimal response in all subgames of the supergame \citep{subgamePerfect}.
\end{itemize}

\paragraph{Related Work}

An introduction to the economic treatment of iterated games can be
found in \citet{iteratedGames}. The fundamental concept of game
theory is that of Nash equilibrium \citep{nash}, which has been adapted
for the study of repeated and dynamic games to the concept of subgame
perfect equilibrium first introduced by \citet{selten}. Significantly
influential work on using logical methods and coalgebraic reasoning in
economics include \citet{rationalEscalation} and
\citet{coalgebraicGames}. Open games are also closely related to the
`partially defined games' of \citet{ramseysTheorem}.

\paragraph{Structure of the paper} Section~\ref{sec:prelim} consists
of preliminaries and a summary of previous work on open games;
Section~\ref{sec:fun} introduces a modality for dealing with subgame
perfection; Section~\ref{sec:2cells} introduces morphisms between
games, and Section~\ref{sec:final} consists of our final coalgebra
semantics for infinite open games. Finally
Section~\ref{sec:conclusion} contains concluding remarks and
discussions of further work.

\section{Preliminaries}
\label{sec:prelim}

The key concept of \citet{compositionalGames} is the following:
\begin{definition}[Open Game]
Let $X$, $Y$, $R$ and $S$ be sets. An \emph{open game} $\G = (\Sigma_\G, P_\G, C_\G, E_\G) : (X,S) \ra (Y,R)$ consists of:
\begin{itemize}
\item a set $\Sigma_\G$ of \emph{strategy profiles},
\item a \emph{play} function $P_\G : \Sigma_\G \ra (X \ra Y)$,
\item a \emph{coutility} function $C_\G : \Sigma_\G \ra (X \times R \ra S)$, and
\item an \emph{equilibrium function} $E_\G :  X \times (Y \ra R) \ra \Pow \Sigma_\G$.
\end{itemize}
\end{definition}

\noindent We sometimes write $\G:(X,S)\raS{\Sigma} (Y,R)$ to make the set of strategies explicit.
Intuitively, the set $X$ contains the states of the game, $Y$ the
moves, $R$ the utilities and $S$ the coutilities. The set
$\Sigma_\G$ contains the strategies we are trying to pick an optimal
one from. The play function $P_\G$ selects a move given a strategy and a state, while the coutility function $C_\G$ computes the coutility extruded from the game, given a strategy, state and utility. Finally, if $\sigma \in E_\G \, x \,k$, then $\sigma$ is an optimal strategy in state $x$ and with utility given by $k:Y\rightarrow R$.
 The main result of~\citet{compositionalGames} can be stated as follows:

\begin{theorem}
  The collection of pairs of sets, with open games
  $\G:(X,S) \rightarrow (Y,R)$ as morphisms, forms a symmetric monoidal
  category $\Open$.\footnote{Actually, one needs to quotient by the
    equivalence relation induced by isomorphism of strategies but we
    simplify presentation here by dealing with representatives directly.}
\end{theorem}
\begin{proof}
  The composition of $\G$ and $\HH$ is given by the game with
  strategies
  $\Sigma_{\HH \circ \G} = \Sigma_{\G} \times \Sigma_{\HH}$, play
  function the composition of the respective play functions from $\HH$
  and $\G$, and coutility function the composition in reverse of the
  coutility functions from $\HH$ and $\G$, using the play function of
  $\G$ to produce a state for $\HH$. Finally
  $(\sigma_1, \sigma_2) \in E_{\HH \circ \G}\,x\,k$ if and only if
  $\sigma_1 \in E_{\G}\, x\, k'$, where
  $k'\,y = C_{\HH}\,\sigma_2\, y\, (k\, (P_{\HH}\,\sigma_2\,y))$, and
  $\sigma_2 \in E_{\HH}\, (P_{\G}\,\sigma'\,x)\, k$ for all
  $\sigma' \in \Sigma_{\G}$.

  The monoidal product is given by Cartesian product in the category
  of sets, with componentwise action on the strategies, play functions
  and coutility functions of open games, and
  $(\sigma_1, \sigma_2) \in E_{\G \otimes H}\,(x_1, x_2)\,k$ if and
  only if
  $\sigma_1 \in E_{\G}\,x_1\,((\pi_1 \circ k)(\_,
  P_{\HH}\,\sigma_2\,x_2))$ and
  $\sigma_2 \in E_{\HH}\,x_2\,((\pi_2 \circ k)(P_{\G}\,\sigma_1\,x_1,
  \_))$. The unit of this monoidal structure is $(\One,\One)$, while the
  symmetry is inherited from the Cartesian product in $\Set$.
\end{proof}

\section{Subgame-Perfection and Conditioning}
\label{sec:fun}

\input{twofunctors.tex}

\section{2-Cells and Coutility Free Games}
\label{sec:2cells}

\input{2cells.tex}

\section{The iterated game as a final coalgebra}
\label{sec:final}

\input{fincoal.tex}

\section{Conclusions and Future Work}
\label{sec:conclusion}

The main contributions of this paper are on the one hand a notion of
morphism between open games and --- based on this notion --- the
representation of the infinite iteration of a given game as a final
coalgebra.  This provides a first extension of the compositionality
results from~\citet{compositionalGames} to infinitely repeated
games. Nevertheless a number of challenges remain: firstly, we need to
extend our construction to state-full games and to games with
non-trivial coutility function.  The former seems straightforward, at
least if we confine ourselves to games that share the same state space
$X$.  Secondly, we need to make the link of our work to
subgame-perfect equilibria more explicit.  Finally, after having
represented infinitely repeated games as final coalgebra, we will be
able to provide new reasoning tools for such games based on
coinduction and coalgebraic logics.

\paragraph{Acknowledgements}
This work was supported by EPSRC grants EP/M016951/1, EP/K023837/1,
and EP/N015843/1.

\bibliographystyle{elsarticle-harv}
\bibliography{openGames}
\newpage
\appendix
\input{appendix.tex}

\end{document}

%% file: macros.tex
\newcommand{\Open}{\mathsf{Open}}
\newcommand{\TwoOpen}[1]{2\Open_{#1}}

\newcommand{\isdef}{\mathrel{:=}}
\newcommand{\Set}{\mathsf{Set}}
\newcommand{\id}{\mathrm{id}}
\newcommand{\now}{\mathrm{now}}
\newcommand{\ltr}{\mathrm{ltr}}
\newcommand{\nil}{\epsilon}
\newcommand{\hd}{\mathrm{hd}}
\newcommand{\tl}{\mathrm{tl}}
\newcommand{\cons}{\mathrel{:\mkern-1.5mu :}}
\newcommand{\bra}[1]{\left\langle #1 \right\rangle}
\newcommand{\brak}[1]{\left( #1 \right)}

\newcommand{\equilib}{E}
\newcommand{\iter}[1]{{#1}_{\G_\omega}}
\newcommand{\Pow}{\mathcal{P}}

\newcommand{\G}{\mathcal{G}}
\newcommand{\HH}{\mathcal{H}}

\newcommand{\One}{\mathbf{1}}

\newcommand{\equiop}{\Phi}
\newcommand{\unfo}{\mathrm{unf}}

\newcommand{\kmap}[1]{{#1}^*}

\newcommand{\ra}{\rightarrow}
\newcommand{\raS}[1]{\xrightarrow{#1}}


%% file: twofunctors.tex
Intuitively, we play two rounds of a game by composing the game with
itself. However, this is not quite right: in the composite game
$\Sigma_{\HH \circ \G} = \Sigma_\HH \times \Sigma_\G$, and thus the
second game $\HH$ cannot react to the moves played by the first game
$\G$. This clearly does not match practice. Rather than introduce a
new form of composition, we introduce a modality which allows us to
condition a game to react to every possibility in some set $A$.

\begin{definition}
  Let $A$ be a set. Given a game $\HH:(X,\,S)\raS{\Sigma} (Y,R)$, we
  define the game
  $ A\ra \HH : (A\times X,\,S) \raS{A \ra \Sigma} (A \times Y,\,R) $ by
\begin{itemize}
\item the play function
  $ P_{A \ra \HH}\,(a,x)\,(f:A\ra \Sigma_{\HH}) \,=\,(a,\,P_{\HH}\,x\,(fa))$
\item the coutility function
  $ C_{A \ra \HH}\, (a,x) \,f\,r\,=\,C_{\HH}\,x\,(fa)\, r$
\item the equilibrium function
\[
 f\in E_{A \ra \HH}(a,x)\,(k:A\times Y \ra R)  \;\; \mbox{ iff }\;\; (\forall a' \in A) \,\, fa' \, \in \, E_{\HH}\, x\,\, k(a',\_)
\]
\end{itemize}
\end{definition}

Note how a strategy in $A \ra \HH$ is a set of strategies, one for
each element of $A$, and that for a strategy $f$ to be optimal in
$A \rightarrow \HH$, each of its components must be optimal in
$\HH$. This captures the notion of subgame-perfection. Clearly we
have:

\begin{lemma}
  The mappings $(X,S) \mapsto (A\times X,S)$ and
  $\HH \mapsto A \rightarrow \HH$ define a functor
  $A \rightarrow \_ : \Open \ra \Open$. \qed
\end{lemma}


%% file: 2cells.tex
Fundamentally, if we have a game $\G:(X,S) \raS{\Sigma} (Y,R)$, its
infinite iteration $\G_\omega$ will be constructed compositionally as
the final coalgebra of the functor $F_\G$ defined by
$\HH \mapsto (Y \ra \HH) \circ \G$. However, this means that games
will acquire universal properties and thus we need a notion of
morphism between games.  Further, $\G_\omega$ will satisfy
$\G_\omega \cong (Y \ra \G_\omega) \circ \G$, and hence the equation
\[
C_{\G_\omega} \, x \, \sigma \, r = C_\G \, x \, \sigma_0 \, (C_{\G_\omega} \, x' \, \sigma_1 \, r)
\]
relating coutility of $\G_\omega$ and coutility of $\G$ must
hold. Here, the strategy $\sigma$ for $\G_\omega$ decomposes into
$\sigma_0$ for the first round and $\sigma_1$ for later rounds, and
$x'$ is the state after the first round is completed. This equation
does not always have a unique solution --- for instance if
$C_\G \, x \, \sigma \, r = r$.  Hence, to recover uniqueness, we
restrict to games $\G$ where $C_\G \, x \, \sigma \, r = r$ in this
paper. This is not a great restriction as in standard game theory
there is no coutility. For the sake of presentation, we will also only
consider state free games. Next, for $F_\G$ to type check, the type of
utility and coutility of $\G$ must be the same, and thus we fix some
set $R$ and only consider games whose utility and coutility is $R$. To
summarise, in this paper we consider open games
$\G : (\One, R) \raS{\Sigma} (Y, R)$ with state $\One$, utility and
coutility the set $R$, and coutility function $C\,\sigma\,r = r$. We
define morphisms between such games as follows:

\begin{definition}
  Let $R$ be a set. Given two games
  $\G: (\One, R) \raS{\Sigma} (Y, R)$ and
  $\G':(\One, R) \raS{\Sigma'} (Y', R)$, a morphism
  $\alpha: \G \ra \G'$ consists of a pair of functions
  $\alpha = {(\alpha_{Y}:\,Y \ra \, Y', \alpha_{\Sigma}:\,\Sigma \ra
    \Sigma')}$ such that
\begin{enumerate}[(i)]
\item $\alpha_{Y}(P \, \sigma)\,=\,P' \, (\alpha_{\Sigma} \, \sigma)$, and
\item for every $\sigma \in \Sigma$ and $k:Y'\ra R$,
if $ \sigma \in E \,(k \circ \alpha_Y)$ then $\alpha_{\Sigma}(\sigma)\in E' \, k$.
\end{enumerate}
\end{definition}

We trust the reader will not be confused by the fact that games are
morphisms in $\Open$ but also have morphisms between them --- this
simply reflects inherent 2-categorical structure.  The category whose
objects are open games $\G : (\One, R) \raS{\Sigma} (Y, R)$ for some
$\Sigma$, $Y$ (and a fixed $R$), and whose morphisms are the morphisms
between such open games is denoted $\TwoOpen{R}$.  We are now in position
to define the functor $F_\G:\TwoOpen{R} \ra \TwoOpen{R}$ whose final
coalgebra will be the infinite iteration of the game $\G$.

\begin{theorem}
  Let $R$ be a set and $\G : (\One, R) \raS{\Sigma} (Y, R)$. The
  mapping $F_\G \HH = (Y \ra \HH) \circ \G$ defines a functor
  $F_\G : \TwoOpen{R} \ra \TwoOpen{R}$.
\end{theorem}
\begin{proof}
Given a morphism $\alpha: \HH \ra \HH'$, we define
$F_\G(\alpha) :  F_\G \HH \ra F_\G \HH'$ by
\[
(F_\G(\alpha))_\Sigma\,(\sigma,f) = (\sigma, \alpha_\Sigma \circ f)
\qquad\qquad
 (F_\G(\alpha))_Y\,(y, z) = (y, \alpha_Y\,z)
\]
The play function and equilibrium preservation conditions are easily
checked.
\end{proof}


%% file: fincoal.tex
From now on, let $R$ be an arbitrary set, used as utility and
coutility for all our games, and write $\TwoOpen{}$ for $\TwoOpen{R}$.

\subsection{Definition of the iterated game}

Let us fix an arbitrary open game $\G : (1,R) \raS{\Sigma} (Y,R)$ that
we want to iterate infinitely often via the final coalgebra of the
functor $F_\G: \TwoOpen{} \ra \TwoOpen{}$ from the previous section, mapping
$\HH: (1,R) \raS{\Sigma_{\HH}} (Y_\HH,R)$ to
${(Y \ra \HH) \circ \G}: (1,R) \raS{\Sigma \times (Y \ra \Sigma_{\HH})}
(Y \times Y_\HH,R)$.  We first describe $F_\G$-coalgebras, then our
candidate $\G_\omega$ for the final $F_\G$-coalgebra, and conclude with a
proof that $\G_\omega$ really is final.  As a first step we need to
recall two endofunctors on the category of sets and their final
coalgebras.

\begin{fact}
  Given two sets $I$ and $O$ we let $D(I,O): \Set \ra \Set$ be the
  functor given by $D(I,O) X \isdef O \times X^I$ and by
  $D(I,O) (f:X \ra Y) \isdef \id_O \times f^I$.  Furthermore, for a
  set $Y$, we define the functor $S(Y): \Set \ra \Set$ by putting
  $S(Y) X = Y\times X$ and $S(Y)(f:X \ra Y) \isdef \id_Y \times
  f$. The final $D(I,O)$-coalgebra is
  \[
    \xymatrix{(I^* \ra O) \ar[rr]^-{\left\langle \now,\ltr\right\rangle}
              &&
              O \times (I^* \ra O)^I}
  \]
  where $\now(f) \isdef f(\nil)$ and $\ltr(f)=\lambda i. \lambda w. f(iw)$
  (cf.~\citet[Ex.~9.5]{rutt00:univ}\footnote{{\em Loc.cit.} proves this for $I = 2$ but the argument can be easily adapted for arbitrary $I$.}). The final $S(Y)$-coalgebra is
  \[
    \xymatrix{Y^\omega \ar[rr]^-{\left\langle \hd, \tl \right\rangle}
              &&
              Y \times Y^\omega}
  \]
  where $Y^\omega$ is the set of infinite streams over $Y$,
  $\hd(y_0 y_1 \dots) \isdef y_0$, i.e., $\hd$ maps a stream to its
  first element (its ``head'') and
  $\tl(y_0 y_1 y_2 \dots) \isdef y_1 y_2 \dots$, i.e., $\tl$ maps a
  stream to its tail (cf.~\citet[Ex.~9.4]{rutt00:univ}).
\end{fact}

The above final coalgebras are fundamental for our representation
of iterated games: The final $S(Y)$-coalgebra consists of
all infinite sequences of moves of the one-round game, while the final $D(Y,\Sigma)$-coalgebra
represents the set of strategies that map lists of moves --- representing moves chosen in previous rounds --- to a strategy for the next round.
As notation, for $\sigma : Y^* \ra \Sigma$ we
abbreviate $\now(\sigma)$ to $\sigma_0$, $\ltr(\sigma)$ to $\sigma'$, and
use $(\cons) : Y \times Y^\omega \ra Y^\omega$ to denote the cons-operator
on lists. Let us now define the $\omega$-iteration of $\G$.

\begin{definition}
  The $\omega$-iteration $\G_\omega : (1,R) \ra (Y^\omega,R)$ of $\G : (1,R) \ra (Y,R)$ has strategies $\iter{\Sigma}$
  given by $\iter{\Sigma} \isdef Y^* \ra \Sigma_\G$, and play
  function $\iter{P}$ given by
  \[
    \iter{P} \sigma = P_\G \sigma_0 \cons \iter{P} (\lambda z. \sigma (P_\G \sigma_0 \cons z))
  \]
 To define the equilibrium function
  $\iter{\equilib}: (Y^\omega \ra R) \ra \Pow \iter{\Sigma}$, we first
  define an operator  $\equiop : (\Pow \iter{\Sigma})^{(Y^\omega \ra R)} \ra (\Pow
  \iter{\Sigma})^{(Y^\omega \ra R)}$ by putting
  \begin{eqnarray*}
    \sigma \in \equiop \, \Gamma \, k & \mbox{ if }
      & \sigma_0 \in \equilib_\G(\lambda y. k(y\cons\iter{P} (\sigma' y))) \\
      & & \text{and} \quad \forall y' \in Y. \; \sigma' y' \in \Gamma(\lambda z. k(y'\cons z))
  \end{eqnarray*}
  Clearly $(\Pow \iter{\Sigma})^{(Y^\omega \ra R)}$
  forms a complete lattice by lifting the complete lattice structure
  of $\Pow \iter{\Sigma}$ pointwise to the function
  space. Furthermore, $\equiop$ is obviously a monotone operator on
  that complete lattice and therefore has a smallest and a greatest
  fixpoint.  We define $\iter{\equilib}$ to be the greatest fixpoint
  of $\equiop$.
\end{definition}

Notice that the above approach means we do not have to fix a
particular utility function $Y^\omega \ra R$ in advance by some
arbitrary form of discounting, but rather work with all possible
utility functions, allowing the user maximum flexibility.

\begin{lemma}
\label{lem:gfp}
  Let $\sigma \in \iter{\Sigma}$. Then
  \begin{enumerate}[(i)]
  \item for all utility functions $k: Y^\omega \ra R$ we have
    $\sigma \in \iter{\equilib} (k)$ if and only if
    $ \sigma \in \equiop(\iter{\equilib})(k) $, and
  \item for every
    $\Gamma \in (\Pow \iter{\Sigma})^{(Y^\omega \ra R)}$, if
    $\Gamma \leq \equiop (\Gamma)$ then also
    $\Gamma \leq \iter{\equilib}$.
  \end{enumerate}
  Here $\leq$ denotes the order on
  $(\Pow \iter{\Sigma})^{(Y^\omega \ra R)}$ given by
  $\Gamma \leq \Delta$ if $\Gamma(k) \subseteq \Delta(k)$ for all
  $k: Y^\omega \ra R$.
\end{lemma}
\begin{proof}
  The first item follows since $\iter{\equilib}$ is a
  fixpoint of $\equiop$, the second because it is the greatest
  such, thus also the greatest post-fixpoint wrt the order
  $\leq$.
\end{proof}

\subsection{Proof of finality}

In this section we are going to show that $\G_\omega$ is a final coalgebra of the functor
$F_\G = (Y \ra \_) \circ \G: \TwoOpen{} \ra \TwoOpen{}$. We have two
things to show:
\begin{enumerate}[(i)]
\item $\G_\omega$ is an $F_\G$-coalgebra, and
\item for any other $F_\G$-coalgebra $\gamma: \HH \ra F_\G \HH$, there exists
  a {\em unique} $F_\G$-coalgebra morphism
  $\brak{\unfo_\Sigma,\unfo_Y}: \HH \ra \G_\omega$.
\end{enumerate}

The first item is formulated in the following proposition --- its straight\-forward proof can be found in the \hyperlink{proof:coalg}{appendix}.
\begin{prop}
\label{prop:iscoalg}
  The $\omega$-iteration $\G_\omega$ of $\G$ is an $F_\G$-coalgebra with
  coalgebra map $\alpha = \brak{\bra{\now,\ltr},\bra{\hd,\tl}} :
  \G_\omega \ra F_\G \G_\omega$. \qed
\end{prop}

We are now ready to prove that $\G_\omega$ indeed is the final
$F_\G$-coalgebra.  To this end we consider an arbitrary
$F_\G$-coalgebra $\HH$ with coalgebra map
$\brak{\bra{\now_\HH,\ltr_\HH},\bra{\hd_\HH,\tl_\HH}}$. We have to
prove that there is a morphism
$\bra{\unfo_\Sigma,\unfo_Y}: \HH \ra \G_\omega$ such that the following
diagram commutes:
\[
  \xymatrix{
    F_\G \HH \ar@{-->}[rr]^{F_\G \brak{\unfo_\Sigma,\unfo_Y}} & & F_\G \G_\omega \\
    \HH \ar[u]^{\brak{\bra{\now_\HH,\ltr_\HH},\bra{\hd_\HH,\tl_\HH}}} \ar@{-->}[rr]_{\brak{\unfo_\Sigma,\unfo_Y}} & & \G_\omega \ar[u]_{\brak{\bra{\now,\ltr},\bra{\hd,\tl}}}}
\]
It is easy to see that such a $F_\G$-coalgebra morphism --- if it
exists --- must be unique because commutativity of the above diagram implies
commutativity of the following two diagrams in the category of sets:

\begin{equation}\label{equ:setfinal}
\begin{array}{cc}
\xymatrix{\Sigma_\G \times \Sigma_\HH^Y \ar[rr]^{D(Y,\Sigma_\G) (\unfo_\Sigma)} & & \Sigma_\G \times \iter{\Sigma}^Y \\
  \Sigma_\HH \ar[u]_{\bra{\now_\HH,\ltr_\HH}} \ar[rr]_{\unfo_\Sigma} &  & \iter{\Sigma} \ar[u]_{\bra{\now,\ltr}}} &
  \xymatrix{Y \times Y_\HH \ar[rr]^{\id_Y \times (\unfo_Y)} &  & Y \times Y^\omega \\
  Y_\HH \ar[u]_{\bra{\hd_\HH,\tl_\HH}} \ar[rr]_{\unfo_Y} &  &  Y^\omega \ar[u]_{\bra{\hd,\tl}}}
\end{array}
\end{equation}

In other words $\unfo_\Sigma$ and $\unfo_Y$ have to be
$D(Y,\Sigma_\G)$- and $S(Y)$-coalgebra morphisms, respectively, and
these are uniquely determined by the fact that their codomains are
the respective final coalgebras.

This means that to show that $\G_\omega$ is a final $F_\G$-coalgebra, we
have to prove that the pair of functions $\bra{\unfo_\Sigma,\unfo_Y}$
defined via the diagrams in~\eqref{equ:setfinal} is a $F_\G$-coalgebra
morphism. We need several lemmas.

\begin{lemma}\label{lem:playfct}
  For every $\sigma \in \Sigma_\HH$ we have
  $\unfo_Y(P_\HH (\sigma)) = \iter{P}(\unfo_\Sigma (\sigma))$.
\end{lemma}
\begin{proof}
  To see this we define a relation
  \[
    Q \isdef \{ (\unfo_Y(P_\HH (\sigma)),\iter{P}(\unfo_\Sigma
    (\sigma))) \mid \sigma \in \Sigma_\HH \} \subseteq Y^\omega \times
    Y^\omega
  \]
  and we prove that $Q$ is a $S(Y)$-bisimulation, i.e., that for each
  $(\tau_1,\tau_2) \in Q$ we have $\hd (\tau_1) = \hd(\tau_2)$ and
  $(\tl(\tau_1),\tl(\tau_2)) \in Q$. From the coinduction principle it
  follows that any two streams related by $Q$ are equal which implies
  the lemma. The proof that $Q$ is a bisimulation is contained in the
  \hyperlink{proof:bisim}{appendix}.
\end{proof}

We now turn to the verification of the equilibrium condition for
$\brak{\unfo_\Sigma,\unfo_Y}$.  First we use
$\brak{\unfo_\Sigma,\unfo_Y}$ to define an indexed predicate on
$\iter{\Sigma}$ (which can be thought of as the image of $\equilib_\HH$
under $\brak{\unfo_\Sigma,\unfo_Y}$). This predicate will be a post-fixpoint of
$\equiop$ which will then imply the desired
equilibrium condition.

\begin{definition}
\label{def:Ehat}
  We define an indexed predicate
  $\hat{\equilib}_\HH: (Y^\omega \ra R) \ra \Pow \iter{\Sigma}$ by
  putting $\sigma \in \hat{\equilib}_\HH k$ if
  $\exists \sigma' \in \Sigma_\HH$ s.t.\ $\unfo_\Sigma (\sigma')=\sigma$
  and $\sigma' \in \equilib_\HH(k \circ \unfo_Y)$.
\end{definition}

\begin{definition}
  We define a map
  $\kmap{(\_)}:(Y^\omega \ra R) \ra (Y \times Y_\HH \ra R)$ by putting
  $\kmap{k} = \lambda y. \lambda z. k(y \cons \unfo_Y(z))$.
\end{definition}

\begin{lemma}
\label{lem:equilibH}
  For $k: Y^\omega \ra R$ and $\sigma' \in \Sigma_\HH$, if
  $\sigma' \in \equilib_\HH(k \circ \unfo_Y)$, then
  \begin{enumerate}[(i)]
  \item
    $\now_\HH(\sigma') \in \equilib_\G( \lambda
    y. \kmap{k}(y,P_\HH(\ltr_\HH(\sigma')(y))))$, and
  \item for all $y' \in Y$ we have
    $\ltr(\sigma')(y') \in \equilib_\HH(\lambda z.\kmap{k}(y',z))$.
  \end{enumerate}
\end{lemma}
\begin{proof}
  Suppose $\sigma' \in \equilib_\HH(k \circ \unfo_Y)$.
  Observe that  $k \circ \unfo_Y = \kmap{k} \circ \bra{\hd_\HH,\tl_\HH}$, so
  this is equivalent to
  $\sigma' \in \equilib_\HH(\kmap{k} \circ \bra{\hd_\HH,\tl_\HH})$ and ---
  as $\bra{\hd_\HH,\tl_\HH}$ is a morphism of open games --- we obtain
  $\bra{\now_\HH,\ltr_\HH}(\sigma') \in \equilib_{F_\G \HH}(\kmap{k})$. The
  lemma now follows by spelling out the definition of
  $\equilib_{F_\G \HH}(\kmap{k})$.
\end{proof}

We are now ready to prove the key fact that $\hat{\equilib}_\HH$ is a
post-fixpoint of $\equiop$.

\begin{lemma}
\label{lem:postfix}
  Let $\sigma \in \iter{\Sigma}$ be a strategy such that
  $\sigma \in \hat{\equilib}_\HH k$ for some $k: Y^\omega \ra R$ . Then
  $\sigma \in \equiop(\hat{\equilib}_\HH)(k)$.
\end{lemma}

\begin{proof}
  The assumption $\sigma \in \hat{\equilib}_\HH k$ means that there is
  some $\sigma' \in \Sigma_\HH$ such that
  $\unfo_\Sigma (\sigma') = \sigma$ and
  $\sigma' \in \equilib_\HH(k \circ \unfo_Y)$.
%
  We need to verify that
  \begin{enumerate}[(a)]
  \item\label{P1}
    $\now(\sigma) \in \equilib_\G(\lambda y. k(y \cons \iter{P}
    \ltr(\sigma)(y)))$, and
  \item\label{P2}
    for all $y' \in Y$ we have $\ltr(\sigma)(y') \in \hat{\equilib}_\HH(\lambda
    z. k(y' \cons z))$.
  \end{enumerate}
  For (\ref{P1}), note that by the diagram for strategies in
  (\ref{equ:setfinal}) we have
  $\now (\sigma) = \now(\unfo_\Sigma(\sigma')) =
  \now_\HH(\sigma')$. Using the first item of
  Lemma~\ref{lem:equilibH}, we obtain
  \begin{eqnarray*}
    \now(\sigma) & \in &  \equilib_\G( \lambda y. \kmap{k}(y,P_\HH(\ltr_\HH(\sigma')(y))))\\
                 & = & \equilib_\G (\lambda y. k(y \cons\unfo_Y(P_\HH(\ltr_\HH(\sigma')(y))))) \\
                 & \stackrel{\mbox{\tiny Lemma~\ref{lem:playfct}}}{=} &
                                                                        \equilib_\G (\lambda y. k(y\cons\iter{P}(\unfo_\Sigma(\ltr_\HH(\sigma')(y))))) \\
                 & \stackrel{\mbox{\tiny (\ref{equ:setfinal})}}{=} & \equilib_\G (\lambda y. k(y\cons\iter{P}( \ltr
                                                                     (\unfo_\Sigma(\sigma'))(y)))) \\
                 & = & \equilib_\G (\lambda y. k(y\cons\iter{P}( \ltr
                       (\sigma)(y))))
  \end{eqnarray*}
  which establishes~(\ref{P1}).

  For (\ref{P2}), it suffices to define for each $y' \in Y$ a suitable
  strategy $\sigma_{y'}' \in \Sigma_\HH$ such that
  $\unfo_\Sigma \sigma_{y'}' = \ltr(\sigma)(y')$ and
  $\sigma_{y'}' \in \equilib_\HH(\lambda z.k(y'\cons\unfo_Y z))$. We
  claim that for an arbitrary $y' \in Y$ the strategy
  $\sigma_{y'}' \isdef \ltr_\HH(\sigma')(y')$ meets these
  conditions. The first condition is again an easy consequence of
  (\ref{equ:setfinal}) and the fact that
  $\unfo_\Sigma (\sigma') = \sigma$.  For the second condition we note
  that $\sigma_y' \in \equilib_\HH(\lambda z.\kmap{k}(y',z))$ as a
  consequence of $\sigma' \in \equilib_\HH(k \circ \unfo_Y)$ and the
  second item of Lemma~\ref{lem:equilibH}.  The claim follows now from
  $ \sigma_{y'}' \in \equilib_\HH(\lambda z.\kmap{k}(y',z)) =
  \equilib_\HH(\lambda z.(y' \cons \unfo_Y(z)))$.
\end{proof}

We are now ready to prove the main theorem of this section.

\begin{theorem}
  Let $\G:(1,R) \ra (Y,R)$ be an open game and let $\G_\omega$ be its
  $\omega$-iteration.  Then $\G_\omega$ is a final $F_\G$-coalgebra.
\end{theorem}
\begin{proof}
  By our discussion at the beginning of this subsection it suffices to
  show that for an arbitrary $F_\G$-coalgebra
  $(\HH,\brak{\bra{\now_\HH,\ltr_\HH},\bra{\hd_\HH,\tl_\HH}})$ the map
  $\brak{\unfo_\Sigma,\unfo_Y}$ consisting of the coalgebra morphisms
  in~\eqref{equ:setfinal} is a morphism of open
  games. Lemma~\ref{lem:playfct} shows that
  $\brak{\unfo_\Sigma,\unfo_Y}$ satisfies the morphism condition wrt
  play functions. For checking the equilibrium condition consider an
  arbitrary $\sigma' \in \Sigma_\HH$ and a $k:Y^\omega \ra R$ such that
  $\sigma' \in \equilib_\HH(k \circ \unfo_Y)$.  Then clearly we have
  $\unfo_\Sigma (\sigma') \in \hat{\equilib}_\HH(k)$.  As
  $\hat{\equilib}_\HH$ is a post-fixpoint of $\equiop$ by
  Lemma~\ref{lem:postfix}, we have
  $\hat{\equilib}_\HH(k) \subseteq \iter{\equilib}(k)$, and thus
  $\unfo_\Sigma (\sigma') \in \iter{\equilib}(k)$ as required.
\end{proof}


%% file: appendix.tex
\section{Omitted Proofs}

\begin{proof}[\hypertarget{proof:coalg}{Proof of Prop.~~\ref{prop:iscoalg}}]
  The type of $\alpha$ is ok, we need to show that $\alpha$ is a
  morphism of open games.  Firstly we need to check that $\alpha$
  interacts well with the play functions, i.e., we need to check that
  for all $\sigma \in \iter{\Sigma}$ we have
  $\bra{\hd,\tl}(\iter{P} \sigma) = P_{(Y \ra \G_\omega) \circ
    \G}(\bra{\now,\ltr}(\sigma))$.  This is routine.

  Secondly we need to verify that $\alpha$ satisfies the equilibrium
  condition, i.e., we need to check that for all
  $\sigma \in \iter{\Sigma}$ and all $k:Y \times Y^\omega \ra R$ we
  have that if $\sigma \in \iter{\equilib}(k \circ \bra{\hd,\tl})$ then
  $\bra{\now,\ltr}(\sigma) \in \equilib_{(Y \ra \G_\omega) \circ
    \G}(k)$. To see this consider arbitrary $\sigma$ and $k$ with
  $\sigma \in \iter{\equilib}(k \circ \bra{\hd,\tl})$.  By
  Lemma~\ref{lem:gfp} we have
  $\now(\sigma) \in \equilib_\G(\lambda y. k(y\cons\iter{P}
  \ltr(\sigma)(y)))$ and
  $\forall y' \in Y. \; \ltr(\sigma)(y') \in \iter{\equilib} (\lambda
  z. k(y'\cons z))$ which is equivalent to
  $\bra{\now,\ltr}(\sigma) = \bra{\now(\sigma),\ltr(\sigma)} \in
  \equilib_{(Y \ra \G_\omega) \circ \G}(k)$ as required.
\end{proof}

\begin{proof}[\hypertarget{proof:bisim}{Proof of Prop.~\ref{lem:playfct}}]
  We only show that the relation $Q$ as defined on page~\pageref{lem:playfct} is a bisimulation.
  To this aim we consider an arbitrary
  $\sigma \in \Sigma_\HH$. We calculate
  \begin{eqnarray*}
    \bra{\hd,\tl}(\unfo_Y(P_\HH(\sigma))) & = & S(Y)(\unfo_Y)(\bra{\hd_\HH,\tl_\HH}P_\HH(\sigma)) \\
                                        & = &  S(Y)(\unfo_Y) \left( P_{F_\G \HH} (\bra{\now_\HH,\ltr_\HH}(\sigma))\right) \\
                                        & = &  S(Y)(\unfo_Y) \left(  P_\G \now_\HH(\sigma) \cons P_\HH (\ltr_\HH(\sigma) (P_\G \now_\HH(\sigma))\right) \\
                                        & = & \bra{P_\G \now_\HH(\sigma), \unfo_Y(P_\HH (\ltr_\HH(\sigma) (P_\G \now_\HH(\sigma)))} \\
                                        & = & \bra{P_\G \now (\unfo_\Sigma(\sigma)), \unfo_Y(P_\HH(\ltr_\HH(\sigma) (P_\G \now_\HH(\sigma))))}
  \end{eqnarray*}
  and
  \begin{eqnarray*}
    \bra{\hd,\tl}(\iter{P}(\unfo_\Sigma (\sigma))) & \stackrel{\mbox{\tiny Prop.~\ref{prop:iscoalg}}}{=} & P_{F_\G \G_\omega}(\bra{\now,\ltr}(\unfo_\Sigma (\sigma)))\\
                                                   & \stackrel{\mbox{\tiny Def.}}{=} &  \bra{P_\G \now(\unfo_\Sigma(\sigma)),\iter{P} (\ltr(\unfo_\Sigma(\sigma)(P_\G \now(\unfo_\Sigma(\sigma))))} \\
                                                   & \stackrel{\mbox{\tiny \eqref{equ:setfinal}}}{=} & \bra{P_\G \now(\unfo_\Sigma(\sigma)),\iter{P} (\unfo_\Sigma(\ltr_\HH(\sigma) (P_\G \now_\HH(\sigma))))}.
  \end{eqnarray*}
  This implies that
  $\hd (\unfo_Y(P_\HH(\sigma))) = \hd (\iter{P}(\unfo_\Sigma (\sigma)))$
  and that
  \[
    (\tl (\unfo_Y(P_\HH(\sigma))),\tl(\iter{P}(\unfo_\Sigma (\sigma)))) \in Q
  \]
  as required.
\end{proof}
